\documentclass{amsart}

\newif\ifdraft\draftfalse
\newif\ifcite\citefalse
\newif\ifblow\blowtrue

\usepackage{amsfonts,amssymb, amsmath}
\usepackage{oldgerm}
\usepackage{amscd}
\usepackage{eucal} 
\usepackage{mathrsfs} 
\usepackage[hyperindex]{hyperref}  
\usepackage[alphabetic]{amsrefs}
\ifdraft\ifcite\usepackage{showkeys}\else\usepackage[notcite,notref]{showkeys}\fi\fi

\usepackage{xypic}
\newdir{ >}{{}*!/-5pt/\dir{>}}

\newtheorem{theorem}[equation]{Theorem}

\theoremstyle{remark}

\theoremstyle{definition}
\newtheorem{definition}[equation]{Definition}

\theoremstyle{remark}

\newtheorem{remark}[equation]{Remark}

\newtheorem{rem}[equation]{Remark}

\newtheorem{example}[equation]{Example}

\numberwithin{equation}{section}

\def\bc{\begin{cases}}
\def\ec{\end{cases}}

\def\a{\alpha}

\def\bc{{\mathbb C}}

\def\bu{{\mathbb U}}

\def\er{\eqref}

\def\bc{\mathbb C}

\def\os{\overset}

\def\lp2{L_pH_{2p}}
\def\bean{\begin{eqnarray}}
\def\eean{\end{eqnarray}}
\def\bea{\begin{eqnarray*}}
\def\eea{\end{eqnarray*}}
\def\beq{\begin{equation}}
\def\eeq{\end{equation}}
\def\beq*{\begin{equation*}}
\def\eeq*{\end{equation*}}
\def\bal{\begin{align*}}
\def\eal{\end{align*}}
\def\baln{\begin{align}}
\def\ealn{\end{align}}

\def\beg{\begin{gather*}}
\def\eng{\end{gather*}}
\def\bqu{\begin{question}}
\def\equ{\end{question}}

\def\ban{\begin{proof}[Answer]}
\def\ean{\end{proof}}

\def\p{\partial}

\def\on{\operatorname}

\def\bqu{\begin{question}}
\def\equ{\end{question}}

\def\0110{\begin{matrix} 0 & 1\\1&0\end{matrix}}

\def\fg{\mathfrak{g}}
\def\fh{\mathfrak{h}}

\def\fl{\mathfrak{l}}

\def\fn{\mathfrak{n}}
\def\fo{\mathfrak{o}}

\def\fs{\mathfrak{s}}

\def\ban{\begin{proof}[Answer]}
\def\ean{\end{proof}}

\def\ben{\begin{equation}}
\def\een{\end{equation}}

\def\j1{{(j+1)}}

\def\e{\epsilon}

\def\bu{{\bf u}}

\begin{document}

\title[On characteristic integrals of Toda field theories]{On characteristic integrals of Toda field theories}

\author{Zhaohu Nie}
\email{zhaohu.nie@usu.edu}
\address{Department of Mathematics and Statistics, Utah State University, Logan, UT 84322-3900}



\begin{abstract}
Characteristic integrals of Toda field theories associated to general simple Lie algebras are constructed using systematic techniques, and complete mathematical proofs are provided. Plenty of examples illustrating the results are presented in explicit forms. 
\end{abstract}

\maketitle

\section{Introduction}

First consider the famous Liouville equation, for which we take the following version:
\begin{equation}\label{liouville}
u_{xy}=-e^{2u},
\end{equation}
where $x$ and $y$ are the independent variables, and $u=u(x,y)$ is an unknown function. 
Let
\begin{equation}\label{cl for l}
I=u_{xx}-u_x^2.
\end{equation}
Then
it is easily checked that 
$I_{y}=\tfrac{\p}{\p y}I=0$
for a solution $u(x,y)$ to \er{liouville}. 
$I$ is thus called a \emph{characteristic integral} of \er{liouville}. 

Toda field theories are generalizations of the Liouville equation \er{liouville}. Let $\fg$ be a simple Lie algebra of rank $n$ with Cartan matrix $A=(a_{ij})_{i,j=1}^n$. For $1\leq i\leq n$, let $u^i=u^i(x,y)$ be $n$ unknown functions of the independent variables $x$ and $y$. The Toda field theory associated to $\fg$ 
is the system of nonlinear PDEs:
\begin{equation}\label{toda}
u^i_{xy}=-e^{\rho_i}:=-\exp\bigg({\sum_{j=1}^n a_{ij}u^j}\bigg),\quad 1\leq i\leq n.
\end{equation}
The Liouville equation \er{liouville} is the Toda field theory associated to $A_1=\fs\fl_2$. 

Toda field theories are fundamental integrable systems with rich properties and important applications in mathematics and physics. They have been extensively studied, and we refer to the two books \cites{BBT,LS-book} for surveys on them. In this paper, we are concerned with the explicit forms of their characteristic integrals. 

\begin{definition}
A differential polynomial in the $u^i(x,y)$ for $1\leq i\leq n$ is a polynomial in the $u^i$ and their partial derivatives with respect to $x$ of various orders $u^i_x, u^i_{xx},\cdots$. 

A \emph{characteristic integral} of the Toda field theory \er{toda} is a differential polynomial $I$ in the $u^i(x,y)$ for $1\leq i\leq n$ such that $I_y=0$ for the solutions $u^i(x,y)$ to \er{toda}. 
\end{definition}

By the symmetry between $x$ and $y$ in the Toda field theory \er{toda}, there are differential polynomials $\tilde I$ in the $u^i$ and their $y$-partial derivatives such that $\tilde I_x=0$ but they clearly follow the same patterns as those in the above definition.

Since a differential polynomial in characteristic integrals is another such integral, the characteristic integrals form a differential algebra. Shabat and Yamilov \cite{preprint} showed that the characteristic integrals of the Toda field theory \er{toda} form a \emph{polynomial} differential algebra generated by $n$ \emph{primitive characteristic integrals}, and that the system \er{toda} has such a complete set of characteristic integrals if and only if the matrix $A=(a_{ij})$ is equivalent to the Cartan matrix of a simple Lie algebra.  

For a differential mononomial in the $u^i$, we call by its \emph{degree} 
the sum of the orders of differentiation multiplied by the algebraic degrees of the corresponding factors. Therefore the $I$ in \er{cl for l} has a homogeneous degree $2$. The primitive characteristic integrals can be chosen to be homogeneous and their degrees are equal to the degrees of the Lie algebra $\fg$. Recall that the algebra of adjoint-invariant functions on $\fg$ is a polynomial algebra on $n$ homogenous generators, whose degrees we call the degrees of $\fg$ \cite{K1}. For a proof of this structure theorem, see \cite{Shabat} and \cite{FF-LNM}*{Theorem 2.4.10 and Proposition 2.4.7}. 
This paper is concerned with constructing these primitive characteristic integrals. 

Many works \cite{LS-book, W-sym, gen-W-sym, dem} have been devoted to the characteristic integrals, 
very often under different names such as \emph{local conservation laws, chiral currents, intermediate integrals} or \emph{$W$-algebras} and from different viewpoints. In particular, the works of Guryeva and Zhiber \cite{G-Zhiber, G-Z-arxiv} have shown that the Toda field theories \er{toda} for classical and exceptional Lie algebras are of Liouville type, and have obtained explicit formulas for the generalized Laplace invariants for these systems, which can be used to construct the characteristic integrals and higher symmetries. 

However to this author, the construction results about characteristic integrals for the Toda field theories \er{toda} are 
not explicit enough, in terms of both the formulas and the proofs. Therefore we would like to present here systematic and concrete techniques for constructing the primitive characteristic integrals together with new and self-contained proofs. 

Furthermore, many works \cite{LS-book,ho,ZS} utilize the characteristic integrals to obtain explicit solutions to the original Toda field theories \er{toda} by the method of Darboux. Recently, Anderson, Fels and Vassiliou \cite{ian's} developed a very general approach to the study of Darboux integrable systems. 

By constructing enough characteristic integrals, this paper explicitly shows that the Toda field theories \er{toda} are Darboux integrable. In sequels to this work, Anderson and the author will apply the characteristic integrals found in this paper and the structure theorem in \cite{ian's} to realize the Toda field theory \er{toda} as the quotient of two standard differential systems  \cite{Yamaguchi} depending on $x$ and $y$ separately. We will also apply the method and result of this paper to the setting of differential invariants for the standard differential systems, generalizing work by Mari Beffa \cite{beffa:2008a}. 

This paper is organized as follows. In Section 2, we present our main results. Our main theorem is Theorem \ref{ds gauge}, which is applicable to all simple Lie algebras. It employes the zero curvature representation \cite{LS} of the Toda field theories \er{toda} under a Drinfeld-Sokolov gauge \cite{DS}. When the Lie algebra $\fg$ has a 
non-branching representation (see \er{down}), 
which is the case for Lie algebras of types $A,B,C$ and $\fg_2$, 
we present a more concrete formula in Theorem \ref{quick}, which is easier to use.
We stress that in both cases, we provide complete and novel mathematical proofs. 

In Sections 3, we demonstrate the more direct method in Theorem \ref{quick} in all the applicable cases. 
In Section 4, we illustrate the general method in Theorem \ref{ds gauge} by the example of the characteristic integral related to the Pffaffian of $D_4$. We stress that our formulas are easily implemented using mathematical softwares such as Maple, which this author and Anderson have done. Maple programs and formulas for the characteristic integrals can be found at the DifferentialGeometry Software Project website at the Digital Commons of Utah State Univeristy (http://digitalcommons.usu.edu/dg/). 
This makes our results readily usable for various purposes. 

\smallskip

\noindent{\bf Acknowledgment.} The author thanks Ian Anderson for introducing this topic to him. He thanks Anderson and Pawel Nurowski for helps with Maple. The author also thanks L\'aszl\'o Feh{\'e}r and Luen-Chau Li for their interests and useful correspondences. He thanks the referee for detailed comments and some references.

\section{Main results}

Let us first introduce some terminology. Let $\fh\subset \fg$ be a Cartan subalgebra, and we denote the corresponding set of roots of $\fg$ by $\Delta$, the sets of positive/negative roots by $\Delta_\pm$, and the set of positive simple roots by $\pi=\{\a_i\}_{i=1}^n$. Let $\fg=\fh\oplus \bigoplus_{\a\in \Delta} \fg_\a$ be the root space decomposition. 

For a root $\a$, define its height by $\on{ht}(\a)=\sum_{i=1}^n c_i$ if $\a=\sum_{i=1}^n c_i \a_i$. Also define the principal height gradation
\begin{equation}\label{height}
\fg=\bigoplus_{k=-m}^m \fg_k, \quad \fg_k=\bigoplus_{\on{ht}(\a)=k} \fg_\a,\ \fg_0=\fh,
\end{equation}
where $m$ is the maximal height of the roots. 
We also denote by $\fn=\bigoplus_{\a\in \Delta_+} \fg_\a=\bigoplus_{k>0} \fg_k$ the maximal nilpotent subalgebra, and by $N$ the corresponding unipotent group. 

For $\a\in \Delta_+$, let $e_\a$ and $e_{-\a}$ be root vectors in the root spaces $\fg_\a$ and $\fg_{-\a}$ such that for $H_\a=[e_\a,e_{-\a}]\in \fh$, we have $\a(H_\a)=2$. 
Then the Cartan matrix $A=(a_{ij})_{i,j=1}^n$ of $\fg$ is 
defined by $a_{ij}=\a_i(H_{\a_j})$.

Let us recall the zero curvature representation of \er{toda} following \cite{LS}. 
Let 
\begin{equation}\label{players}
{\bf u}=\sum_{i=1}^n u^i_x H_{\a_i}, \quad
\e=\sum_{i=1}^n e_{-\a_i}, \quad
Y=\sum_{i=1}^n e^{\rho_i} e_{\a_i},
\end{equation}
where as in \er{toda} $\rho_i=\sum_{j=1}^n a_{ij} u^j$. Then the Toda field theory \er{toda} is equivalent to the following zero curvature equation
\begin{equation}\label{zero curv}
[-\p_x+\e+{\bf u},\, \p_y+Y]=0.
\end{equation}

Now let us recall the definition of a Kostant slice $\fs\subset \fg$ \cite{K2}, which is used in a Drinfeld-Sokolov gauge \cite{DS}. Let $\fs$ be a complement of $[\e,\fg]$ in $\fg$, that is, 
\begin{equation}\label{split}
\fg\cong \fs\oplus [\e,\fg].
\end{equation}
Then by \cite{K2}, $\fs\subset \fn$, and $\dim(\fs)=n$ is equal to the rank. We call $\fs$ a Kostant slice, and let $\{s_j\}_{j=1}^n$ be a homogeneous basis of $\fs$ with respect to the height gradation \er{height}. 

By \cite{DS}, we can bring the first element in \er{zero curv} into its Drinfeld-Sokolov gauge. More precisely, there exists an element $g\in N$ (whose components are differential polynomials of the $u^i$) such that 
\begin{equation}\label{bring up}
g (-\p_x+\e+{\bf u}) g^{-1}=-\p_x + \e + {\bf I},\quad {\bf I}=\sum_{j=1}^n I_j s_j\in \fs, 
\end{equation}
where the components $I_j$ are differential polynomials of the $u^i$. 

\begin{theorem}\label{ds gauge} For the solutions $u^i$ to \er{toda}, the differential polynomials $I_j$ for $1\leq j\leq n$ defined in \er{bring up} through a Drinfeld-Sokolov gauge are primitive characteristic integrals, that is, $\p_y I_j=0$.  
\end{theorem}

\begin{proof}
Suppose that for the $g\in N$ in \er{bring up}, we have 
$$g^{-1}(\p_y + Y)g=\p_y + \tilde Y.$$ 
Since $g\in N$ and $Y\in \fg_1$ by \er{players}, we have $\tilde Y\in \fn$. Suppose $\tilde Y=\sum_{i=1}^p Y_i$ by the height decomposition \er{height}, where $p$ is the biggest height of $\fg$. 

By the invariance of the zero-curvature equation \er{zero curv} under the  adjoint action by $g$, from \er{bring up} we get 
\begin{equation}\label{conjged}
\bigg[-\p_x + \e + \sum_{j=1}^n I_j s_j, \p_y + \sum_{i=1}^{p} Y_i\bigg]=0.
\end{equation}

We will prove by induction that all the $Y_i=0$. Clearly this then implies that $\p_y I_j=0$ for $1\leq j\leq n$. 

First recall the basic result of Kostant \cite{K1} that $\ker(\on{ad}_\e) \cap (\fh\oplus \fn)=0$. The term on the left of \er{conjged} with height zero is $[\e, Y_1]=0$, which then implies that $Y_1=0$. 

Now assume $i\geq 2$ and that $Y_{j}=0$ for $j\leq i-1$. Then since the $s_j\in \fn$, 
the term on the left of \er{conjged} with height $i-1$ is 
$$
[\e, Y_i] - \sum_{\on{ht}(s_j)=i-1} (\p_y I_j) s_j =0.
$$
By the decomposition \er{split}, we get $[\e, Y_i]=0$ which then implies that $Y_i=0$. 

These $I_j$ are primitive since without further differentiation, we are obtaining the primitive invariant functions on the Lie algebra by \cite{K2}. 
\end{proof}

\begin{remark}\label{DS gauging}
The calculations of $g\in N$ and the $I_j$ in \er{bring up} can be done inductively using the decomposition \er{split} as in \cite{DS, gen-W-sym}. For the reader's convenience, we provide some details. Consider 
$$
g = e^{a_1}\cdots e^{a_m}, \quad a_i\in \fg_i,\ i=1,\dots,m.
$$
The formula \er{bring up} becomes 
$$
\partial_x g \cdot g^{-1} + g (\e + {\bf u}) g^{-1} = \e + {\bf I}.
$$
Inductively for $i\geq 1$, the $i-1$th component of the left hand side can be uniquely written as $[\e, a_i] + J_{i-1}$ with $J_{i-1} \in \fs\cap \fg_{i-1}$. Then exactly this $a_i$ will do the job of the Drinfeld-Sokolov gauging at grade $i-1$. The algorithm is easily implemented on a mathematical software such as Maple, as long as we know the structure equations of the Lie algebra. This author has worked out the example of $E_6$ Toda field theory by this method. In Section 4, we will see the example of $D_4$. 
\end{remark}

Furthermore very often there are more direct formulas for the characteristic integrals as we explain now. 

Take an irreducible representation $\phi:\fg\to \text{End}\, V$. Let the $\beta_k\in \fh^*$ for $1\leq k\leq m$ be the weights of $\phi$, and 
\begin{equation}\label{wt decomp}
V=\bigoplus_{k=1}^m V_{\beta_k},
\end{equation}
the weight space decomposition. 
We assume that $\dim V_{\beta_k}=1$ and also that the representation $\phi$ does not branch. That is, for each weight $\beta_k$ there is one unique negative simple root $-\a_{i_k}$ such that $\beta_k-\a_{i_k}$ is another weight of $\phi$. 
Order our weights such that $\beta_{k+1}=\beta_k-\a_{i_k}$ for $1\leq k\leq m-1$ and 
we draw the following weight diagram
\begin{equation}\label{down}
\beta_m\os{-\a_{i_{m-1}}}\longleftarrow \cdots \os{-\a_{i_2}}\longleftarrow \beta_2\os{-\a_{i_1}}\longleftarrow \beta_1. 
\end{equation}
Such non-branching representations occur for the first fundamental representations of the Lie algebras $A_n, B_n, C_n$ and $\fg_2$. 

\begin{theorem}\label{quick} If a non-branching representation $\phi$ as above exists with weight diagram \er{down}, then we have 
\begin{equation}\label{Pawel said}
[(\p_x - \beta_{{1}}(\bu))(\p_x - \beta_{2}(\bu))\cdots(\p_x - \beta_{m}(\bu)),\p_y]=0,
\end{equation}
for a $\bu$ \er{players} satisfying \er{zero curv} or equivalently the Toda filed theory \er{toda}. 
The product is in the sense of composition for operators on functions of $x$ and $y$, and we strictly follow the order. If by the Leibniz rule, we expand 
$$
(\p_x - \beta_{{1}}(\bu))(\p_x - \beta_{2}(\bu))\cdots(\p_x - \beta_{m}(\bu))=\p_x^{m}+\sum_{j=1}^{m} I_j \p_x^{m-j},
$$
then \er{Pawel said} implies that 
\begin{equation}\label{reason}
\p_y I_j=0,\quad 1\leq j\leq m. 
\end{equation}
\end{theorem}

\begin{proof}
In the weight decomposition \er{wt decomp}, we let $v_{\beta_1}\in V_{\beta_1}$ be a weight vector for the highest weight, and inductively we define the other weight vectors by $v_{\beta_k}=\phi(e_{-\a_{i_{k-1}}})v_{\beta_{k-1}}$ for $2\leq k\leq m$ by \er{down}. Therefore by \er{players}
\begin{equation}\label{by ep}
\phi(\e) v_{\beta_{k-1}}=v_{\beta_{k}},\quad 2\leq k\leq m.
\end{equation}

The zero curvature equation \er{zero curv} is the compatibility condition for the following system of equations. 
Let ${\psi}(x,y)=\sum_{k=1}^m \psi_k(x,y) v_{\beta_k}$ be a function of $x$ and $y$ with values in $V$. 
Then \er{zero curv} implies that the following system of equations has solutions 
\begin{equation}\label{wave eqn}
\begin{cases}
(-\p_x+\phi({\bf u}+\e))\psi=0\\
(\p_y+\phi(Y))\psi=0.
\end{cases}
\end{equation}

Then by \er{by ep} and $\bu\in \fh$ \er{players}, the first equation, at the weight vector $v_{\beta_k}$, means that 
\begin{equation}\label{induction}
(\p_x  - \beta_k (\bu)) \psi_k=\psi_{k-1},\quad 2\leq k\leq m.
\end{equation}
When $k=1$, we actually have 
$$
(\p_x-\beta_1 (\bu))\psi_1=0,
$$
since $\beta_{1}$ is the highest weight. 
Therefore combining them, we have 
$$
(\p_x - \beta_{1}(\bu))(\p_x - \beta_{2}(\bu))\cdots(\p_x - \beta_{m}(\bu)) \psi_m=0.
$$
On the other hand, the second equation in \er{wave eqn}, at the lowest weight vector $v_{\beta_m}$, quickly gives that 
$$
\p_y \psi_m=0,
$$
since $Y\in \fn$ \er{players}. 
Therefore the above two equations have a solution $\psi_m=\psi_m(x,y)$, which is general. The implied compatibility condition is exactly \er{Pawel said}. The rest of the theorem is clear. 
\end{proof}

We note that again the summands of the $I_j$ with only derivatives of order 1 again are just the algebraic invariant functions. Therefore if we know the degrees of the Lie aglebras, by the structure theorem we know the primitive characteristic integrals. 







\section{Examples for Theorem \ref{quick}}

The first fundamental representations of the Lie algebras $A_n, B_n, C_n$ and $\fg_2$ are non-branching representations \er{down}. Keeping our spirit of being as explicit as possible, we present the formulas in these cases. In this section we follow \cite{FH} for notation and choices of root vectors. The Cartan subalgebras $\fh$ always consist of diagonal matrices. We let $L_i\in \fh^*$ denote the linear function of taking the $i$th element on the diagonal. We also let $E_{ij}$ denote the matrix with a 1 at the $(i,j)$-position and zero everywhere else. For simplicity, we write $\p$ for $\p_x$.

\begin{example}[$A_n$] The simple roots for $A_n$ are $\a_i=L_i-L_{i+1}$ and 
the $H_{\a_i}=E_{i,i}-E_{i+1,i+1}$, $1\leq i\leq n$. 
The Cartan matrix is 
$$
\begin{pmatrix}
2 & -1 & & &\\
-1 & 2 & -1 & & \\
 & \ddots & \ddots & \ddots& \\
 & &-1 & 2 & -1\\
 & & & -1 & 2
 \end{pmatrix}
$$
The degrees of $A_n$ (that is, the degrees of primitive adjoint-invariant functions) are $2,3,\cdots,n+1$. 
The weight diagram for the first fundamental representation is 
$$
L_{n+1}\os{-\a_n}\longleftarrow \cdots \os{-\a_2}\longleftarrow L_2\os{-\a_1}\longleftarrow L_1.
$$
The $\bu$ in \er{players} is 
$$
\bu=\on{Diag}(u^1_x, u^2_x - u^1_x, \cdots, u^n_x - u^{n-1}_x, -u^n_x)
$$
By Theorem \ref{quick}, consider the expansion
\begin{equation}\label{An det}
\begin{split}
&\quad(\p-u^1_x)(\p+u^1_x-u^2_x)\cdots(\p+u^{n-1}_x-u^n_x)(\p+u^n_x)\\
&=\p^{n+1}+\sum_{j=1}^n I_j\p^{n-j}.
\end{split}
\end{equation}
Then the $I_j$ for $1\leq j\leq n$ are 
primitive characteristic integrals of the $A_n$ Toda field theory \er{toda}.

\end{example}

\begin{example}[$C_n$, $n\geq 2$]
The simple roots for $C_n$ are $\a_i=L_i-L_{i+1}$ for $1\leq i\leq n-1$ and $\a_n=2L_n$. Also $H_{\alpha_i}=E_{i,i}-E_{i+1,i+1}-E_{n+i,n+i}+E_{n+i+1,n+i+1}$ for $1\leq i\leq n-1$ and $H_{\alpha_n}=E_{n,n}-E_{2n,2n}$. The Cartan matrix is 
$$
\begin{pmatrix}
2 & -1 & & &\\
-1 & 2 & -1 & & \\
 & \ddots & \ddots & \ddots& \\
 & &-1 & 2 & -1\\
 & & & -2 & 2
 \end{pmatrix}
$$
The degrees of $C_n$ are $2,4,\cdots,2n$. 
The weight diagram for the first fundamental representation  is
$$
-L_1\os{-\a_1}\longleftarrow \cdots \os{-\a_{n-1}}\longleftarrow -L_n\os{-\a_n}\longleftarrow L_n\os{-\a_{n-1}}\longleftarrow \cdots\os{-\a_2}\longleftarrow L_2\os{-\a_1}\longleftarrow L_1.
$$
The $\bu$ in \er{players} is 
$$
\bu=\on{Diag}(u^1_x, u^2_x - u^1_x, \cdots, u^n_x - u^{n-1}_x, -u^1_x, -u^2_x + u^1_x, \cdots, -u^n_x + u^{n-1}_x). 
$$

By Theorem \ref{quick}, consider the expansion
\begin{equation}\label{Cn det}
\begin{split}
&\quad (\p-u^1_x)(\p+u^1_x-u^2_x)\cdots(\p+u^{n-1}_x-u^n_x)\\
&\quad\quad (\p+u^{n}_x-u^{n-1}_x)
\cdots(\p+u^2_x-u^1_x)(\p+u^1_x)\\
&=\p^{2n}+\sum_{j=1}^n I_j\p^{2n-2j}+\sum_{j=1}^{n-1} J_j\p^{2n-2j-1}.
\end{split}
\end{equation}
Then the $I_j$ for $1\leq j\leq n$  are 
primitive characteristic integrals of the $C_n$ Toda field theory, and the $J_j$ are some differential polynomials in them.  

\end{example}

\begin{example}[$B_n, n\geq 2$] The simple roots for $B_n$ are $\a_i=L_i-L_{i+1}$ for $1\leq i\leq n-1$ and $\a_n=L_n$. Also  $H_{\alpha_i}=E_{i,i}-E_{i+1,i+1}-E_{n+i,n+i}+E_{n+i+1,n+i+1}$ for $1\leq i\leq n-1$ and $H_{\alpha_n}=2E_{n,n}-2E_{2n,2n}$. The Cartan matrix is 
$$
\begin{pmatrix}
2 & -1 & & &\\
-1 & 2 & -1 & & \\
 & \ddots & \ddots & \ddots& \\
 & &-1 & 2 & -2\\
 & & & -1 & 2
 \end{pmatrix}
$$
The degrees of $B_n$ are $2,4,\cdots,2n$. 
The weight diagram for the first fundamental representation  is
\begin{equation}\label{bn wt}
-L_1\os{-\a_1}\longleftarrow \cdots \os{-\a_{n-1}}\longleftarrow -L_n\os{-\a_n}\longleftarrow 0\os{-\a_n}\longleftarrow L_n\os{-\a_{n-1}}\longleftarrow \cdots\os{-\a_2}\longleftarrow L_2\os{-\a_1}\longleftarrow L_1.
\end{equation}
The $\bu$ in \er{players} is 
$$
\bu=\on{Diag}(u^1_x, u^2_x - u^1_x, \cdots, 2 u^n_x - u^{n-1}_x, -u^1_x, -u^2_x + u^1_x, \cdots, -2u^n_x + u^{n-1}_x, 0). 
$$

By Theorem \ref{quick}, consider the expansion
\begin{equation}\label{Bn det}
\begin{split}
&\quad (\p-u^1_x)(\p+u^1_x-u^2_x)\cdots(\p+u^{n-1}_x-2u^n_x)\\
&\quad\quad \p(\p+2u^{n}_x-u^{n-1}_x)
\cdots(\p+u^2_x-u^1_x)(\p+u^1_x)\\
&=\p^{2n+1}+\sum_{j=1}^n I_j\p^{2n-2j+1}+\sum_{j=1}^n J_j\p^{2n-2j}.
\end{split}
\end{equation}
Then the $I_j$ for $1\leq j\leq n$  are 
primitive characteristic integrals  of the $B_n$ Toda field theory, and the $J_j$ are some differential polynomials in them.  

\end{example}

\begin{example}[$\fg_2$]
The Cartan matrix is 
$$
\begin{pmatrix}
2 & -1\\
-3 & 2
\end{pmatrix}
$$
For simplicity we call the unknown functions $u^1, u^2$ by $u$ and $v$. 
The Toda system is 
$$
\begin{cases}
u_{xy}=-e^{2u-v}\\
v_{xy}=-e^{-3u+2v}
\end{cases}
$$

Although we can work abstractly with just the above Cartan matrix, we choose to use the embedding of $\fg_2\subset \fs\fo_7=B_3$ such that the two roots are 
$$
\a_1=L_1-L_2,\quad \a_2:=L_2-L_3,
$$
and 
\begin{align*}
H_{\alpha_1}&=\on{Diag}(1,-1,2,-1,1,-2,0)\\
H_{\alpha_2}&=\on{Diag}(0,1,-1,0,-1,1,0). 
\end{align*}
The $\bu$ in \er{players} is 
$${\bf u}= 
\on{Diag}(u_x, -u_x + v_x, 2u_x-v_x, -u_x, u_x - v_x, -2u_x+v_x, 0)
.$$

The first fundamental representation of $\fg_2$ is the restriction of that of $B_3$. Therefore we follow the weight diagram \er{bn wt}. 
The degrees of $\fg_2$ are 2 and 6. 

By Theorem \ref{quick}, consider the expansion
\begin{equation}
\begin{split}
&\quad(\p-u_x)(\p+u_x-v_x)(\p-2u_x+v_x)\\
&\quad\quad \p(\p+2u_x-v_x)(\p-u_x+v_x)(\p+u_x)\\
&=\p^{7}+I_1\p^5+I_2\p+\sum_{j=1}^{3} J_j\p^{5-j}+J_4. 
\end{split}
\end{equation}
Then the $I_1$ and $I_2$ are 
primitive characteristic integrals of the $\fg_2$ Toda field theory, and the $J_j$ are some differential polynomials in them.  
\end{example}

\section{An example of Theorem \ref{ds gauge}}

It can be said that Theorem \ref{quick} is a quick way to compute a Drinfeld-Sokolov gauge \er{bring up} when there is a non-branching representation, and this is the viewpoint in \cite{W-sym}. When such a simple representation does not exist, a Drinfeld-Sokolov gauge can still be computed through an inductive procedure (see \cite{DS, gen-W-sym}). In the particular case of $D_n$, where an explicit matrix presentation of the Lie algebra is easy to write down, one can use a computer algebra system like Maple to solve a Drinfeld-Sokolov gauge easily and therefore obtain the characteristic integrals according to Theorem \ref{ds gauge}. 

In the following example, we will show how this works for $D_4$. First recall that the Cartan matrix of $D_n$ is 
$$
\begin{pmatrix}
2 & -1 & & & & \\
-1 & 2 & -1 & & & \\
 & \ddots & \ddots & \ddots& & \\
 & &-1 & 2 & -1 & -1\\
 & & & -1 & 2 & \\
 & & & -1 &  & 2\\
 \end{pmatrix},
$$
and its first fundamental representation branches with the following weight diagram
\small
\begin{equation}\label{splits}
\xymatrix{
 & & & & L_n\ar[ld]_{-\a_n} & & & & \\
-L_1&-L_2\ar[l]_{-\a_1} &\cdots\ar[l]_{-\a_2} & -L_{n-1}\ar[l]_{-\a_{n-2}} & & L_{n-1}\ar[lu]_{-\a_{n-1}}\ar[ld]^{-\a_n} & \cdots\ar[l]_{-\a_{n-2}} & L_2\ar[l]_{-\a_2} & L_1\ar[l]_{-\a_1}\\
& & & & -L_n\ar[lu]^{-\a_{n-1}} & & & & 
}
\end{equation}
\normalsize
Therefore Theorem \ref{quick} can not be applied to the Lie algebras $D_n$. Nor to $F_4$ or the $E$'s.

\begin{example}[$D_4$] For simplicity, we write $u, v, w, z$ for the unknown functions $u^1,\cdots,u^4$. The Toda field theory \er{toda} for $D_4$ is 
$$
\begin{cases}
u_{xy}=-e^{2u-v}\\
v_{xy}=-e^{-u+2v-w-z}\\
w_{xy}=-e^{-v+2w}\\
z_{xy}=-e^{-v+2z}.
\end{cases}
$$

The degrees of $D_4$ are $2, 4, 4, 6$, with the Pfaffian, the square root of the determinant, being an extra adjoint-invariant function of degree 4. 

In this example, we use the convention that the non-degenerate symmetric matrix preserved by the orthogonal matrices in $SO(8)$ is the anti-diagonal matrix $(\delta_{i, 9-j})_{i,j=1}^8$. Then the matrices in $D_4=\fs\fo_8$ are skew-symmetric with respect to the anti-diagonal. 
Using the notation $u_1=u_x$ and so on, the $\e + \bu$ from \er{players} is 
$$
{\small
\e+{\bf u}=
\begin{pmatrix}
u_1 &  & & & & & & \\
1 & v_1-u_1 &  & & & & & \\
 & 1 & z_1+w_1-v_1 &  & & & & \\
 & & 1 & z_1-w_1 & & & & \\
 & & 1 &  & -z_1+w_1 & & & \\
 & & & -1 & -1 & -z_1-w_1+v_1 & & \\
 & & & & & -1 &  -v_1+u_1& \\
 & & & & & & -1 & -u_1
\end{pmatrix}}.
$$
We choose our slice in \er{bring up} to be 
$$
{\small
\e+{\bf I}=
\begin{pmatrix}
0 &  & & & I_2 & & I_4 & \\
1 & 0 &  & & & I_3 & & -I_4 \\
 & 1 & 0 &  & I_1 & & -I_3 & \\
 & & 1 & 0 & & -I_1 & & -I_2 \\
 & & 1 &  & 0 & & & \\
 & & & -1 & -1 &0  & & \\
 & & & & & -1 &  0& \\
 & & & & & & -1 & 0
\end{pmatrix}}.
$$
A unipotent orthogonal matrix $B\in N_{SO(8)}$ is an upper-triangular matrix with $1$'s on the diagonal and element $b_{i,j}$ at the position $(i,j)$ for $i<j$ satisfying some extra conditions for the orthogonality. In particular we have that $b_{4,5}=0$ and $b_{3,4}+b_{5,6}=0$. (Actually one can write out the $B$ using its affine coordinates.) Solving \er{bring up} by Maple, we get the following expression for the primitive characteristic integrals of the $D_4$ Toda field theory. 
{\allowdisplaybreaks
\begin{align*}
I_1&=-(u_{{2}}+v_{{2}}+z_{{2}}+w_{{2}}-{u_{{1}}}^{2}-{v_{{1}}}^{2}-{z_{{1}}}^{2}-{w_{{1}}}^{2}+u_{{1}}v_{{1}}+z_{{1}}v_{{1}}+
w_{{1}}v_{{1}})\\
I_2&=u_{{4}}+v_{{4}}+2\,w_{{4}}+z_{{2}}{v_{{1}}}^{2}+3\,w_{{2}}
v_{{2}}-w_{{2}}{v_{{1}}}^{2}+{w_{{1}}}^{2}v_{{2}}+v_{{2}}z
_{{2}}\\&
-v_{{2}}{z_{{1}}}^{2}+2\,v_{{2}}u_{{2}}-z_{{2}}u_{{2
}}-w_{{2}}{u_{{1}}}^{2}+z_{{2}}{u_{{1}}}^{2}+v_{{3}}u_{{1
}}\\&
-2\,u_{{3}}u_{{1}}+2\,v_{{3}}w_{{1}}-{w_{{1}}}^{2}u_{{2}
}+w_{{2}}u_{{2}}-4\,{w_{{2}}}^{2}-2\,{v_{{2}}}^{2}\\&
-2\,{u_{{2}}}^{
2}+{z_{{1}}}^{2}u_{
{2}}+w_{{2}}v_{{1}}u_{{1}}-z_{{1}}v_{{1}}u_{{2}}+w_{{1
}}v_{{1}}u_{{2}}\\&
-z_{{2}}v_{{1}}u_{{1}}-2\,v_{{1}}w_{{1}}
v_{{2}}+2\,w_{{2}}v_{{1}}w_{{1}}+2\,z_{{1}}v_{{1}}v_{{2}
}-2\,v_{{1}}z_{{1}}z_{{2}}\\&
+u_{{3}}v_{{1}}-2\,v_{{1}}v_{{3
}}+v_{{1}}z_{{3}}-4\,w_{{1}}w_{{3}}+v_{{1}}w_{{3}}-{z_{{1
}}}^{2}{u_{{1}}}^{2}+{w_{{1}}}^{2}{u_{{1}}}^{2}\\&
-{w_{{1}}}^{2}v
_{{1}}u_{{1}}-z_{{1}}{v_{{1}}}^{2}u_{{1}}-w_{{1}}v_{{1}}
{u_{{1}}}^{2}+{z_{{1}}}^{2}v_{{1}}u_{{1}}+z_{{1}}v_{{1
}}{u_{{1}}}^{2}+w_{{1}}{v_{{1}}}^{2}u_{{1}}\\
I_3 &= 2\,u_{{4}}+v_{{4}}+w_{{4}}+{w_{{1}}}^{2}{u_{{1}}}^{2}-2\,v_{{1}}v_{{
3}}+v_{{1}}z_{{3}}-2\,w_{{1}}w_{{3}}+v_{{1}}w_{{3}}\\&
+u_{{
3}}v_{{1}}+3\,v_{{2}}u_{{2}}-w_{{2}}{u_{{1}}}^{2}+w_{{2}
}u_{{2}}+2\,v_{{3}}u_{{1}}-4\,u_{{3}}u_{{1}}+v_{{3}}w_{{1
}}\\&
-{w_{{1}}}^{2}u_{{2}}+v_{{2}}z_{{2}}-v_{{2}}{z_{{1}}}^{2
}+z_{{2}}{v_{{1}}}^{2}+2\,w_{{2}}v_{{2}}-2\,{v_{{2
}}}^{2}-4\,{u_{{2}}}^{2}\\&
-2\,{w_{{2}}}^{2}-{w_{{1}}}^{2}v_{{1}}
u_{{1}}+w_{{1}}{v_{{1}}}^{2}u_{{1}}+2\,z_{{1}}v_{{1}}v_{{2
}}-2\,v_{{1}}z_{{1}}z_{{2}}+w_{{1}}v_{{1}}u_{{2}}\\&
+w_{{2
}}v_{{1}}u_{{1}}
-w_{{1}}v_{{1}}{u_{{1}}}^{2}-{v_{{1}}}^{2}
u_{{2}}+{z_{{1}}}^{2}w_{{2}}-{z_{{1}}}^{2}{w_{{1}}}^{2}+v_{{
2}}{u_{{1}}}^{2}-z_{{2}}w_{{2}}\\&
+z_{{2}}{w_{{1}}}^{2}
+{z_{{
1}}}^{2}v_{{1}}w_{{1}}-{v_{{1}}}^{2}z_{{1}}w_{{1}}+{w_{{1
}}}^{2}v_{{1}}z_{{1}}-z_{{2}}v_{{1}}w_{{1}}-w_{{2}}v_{{1
}}z_{{1}}\\&
+2\,v_{{1}}u_{{2}}u_{{1}}-2\,v_{{2}}v_{{1}}u_{{1
}}\\
I_4 &= -u_6 -\tfrac{1}{2}\, v_6 -w_6 +\text{a lot of other terms which we omit}
\end{align*}
}
Actually the usual Pfaffian of $\e+{\bf u}$, which is the product of the first 4 diagonal entries, is contained in the characteristic integral $I_2$ as the last terms involving only partial derivatives of the first order.

\begin{remark} In \cite{W-sym}, there is a procedure to apply an integration step in using the first fundamental representation \er{splits} of $D_n$ to get an analogous formula to Theorem \ref{quick}.
That integration step will cause us to lose one characteristic integral, which in the case of $D_4$ is exactly the above $I_2$ corresponding to the Pfaffian. 
\end{remark} 
\end{example}

As stated in Remark \ref{DS gauging}, the Drinfeld-Sokolov gauging can be achieved without using a representation. Intrinsically one can use the adjoint representation constructed from the structure equations of the Lie algebra. This author has done this for the $E_6$ Toda field theory.

\bigskip
\end{document}